\documentclass[12pt,reqno]{amsart}

\newtheorem{theorem}{Theorem}[section]
\newtheorem{lemma}[theorem]{Lemma}

\theoremstyle{definition}

\newtheorem{assumption}[theorem]{Assumption}

\usepackage[margin=1in]{geometry}

\theoremstyle{remark}

\numberwithin{equation}{section}

\begin{document}

\title{Vortices and Cosmic Strings in a Generalized Born-Infeld-Higgs Model}

\author{Kai Shao}
%\address{NYU Shanghai, 1555 Century Avenue, Shanghai 200122, China}

\email{kaishao@nyu.edu}

\begin{abstract}
In this paper, we consider two types of topological solitons in a generalized Born-Infeld-Higgs model. We explore the self-dual structure of the model and prove the existence of planar vortex solutions. Furthermore, we couple the system with the Einstein equations and study the cosmic strings problem over $\mathbb R^{1,1}\times S$, where $S$ is a Riemann surface. We prove the existence of cosmic string solutions when $S$ is noncompact.  We also discuss the decay estimates for vortices and cosmic strings at infinity and show that the minimal energy is quantized and depends on the number of vortices and strings, respectively.

\end{abstract}

\maketitle
\section{Introduction}
The Born-Infeld theory has long been an important nonlinear model in electrodynamics. It was proposed by M. Born and L. Infeld \cite{bf, bf2} to overcome the problem of the infinite self-energy of a point charge. In the model, the action density of the classical Maxwell theory
\begin{equation}
\mathcal L=\frac{1}{2}(E^2-B^2)
\end{equation}
is replaced by
\begin{equation}
\mathcal L=b^2\Big(1-\sqrt{1-\frac{1}{b^2}(E^2-B^2)}\Big)
\end{equation}
where $E$ and $B$ are electric and magnetic fields. The positive parameter $b>0$ is called the Born-Infeld parameter and is assumed to be sufficiently large. The classical Maxwell theory can be recovered from the Born-Infeld theory in the limiting case. The action density of the Abelian Born-Infeld-Higgs model over the $(3+1)$-dimensional Minkowski spacetime is defined to be
\begin{equation}
\mathcal L=b^2\Big(1-\sqrt{1+\frac{1}{2b^2}F_{\mu\nu}F^{\mu\nu}}\Big)+D_{\mu}\phi\overline{D^{\mu}\phi}-V(|\phi|^2)
\end{equation}
where $F_{\mu\nu}=\partial_{\mu}A_{\nu}-\partial_{\nu}A_{\mu}$ is the electromagnetic field induced from the gauge field $A_{\mu}$, $\phi$ is a complex scalar field, $D_{\mu}\phi=\partial_{\mu}\phi-\mathrm{i}A_\mu\phi$ is the gauge-covariant derivative, $V$ is the Higgs potential function, and $\{g_{\mu\nu}\}$ is the Minkowski metric of signature $(+ - - -)$. The Born-Infeld theory has not only been one of the most important nonlinear electromagnetic models but also has close connections to many fields in physics, such as membrane theory \cite{by, gb1}, cosmology, and string theory \cite{gb2, gb3}. The effectiveness of removing singularity in electromagnetism also inspires the application of Born-Infeld type structure in the theory of gravity to avoid the spacetime singularity \cite{bj, yang21, yang22}.  From a PDE point of view, the Born-Infeld model also shows its relevance in maximal hypersurface equations \cite{linyang, yang00} and magnetohydrodynamics equations \cite{by, by2}. Because of its wide connection to mathematics and physics, the model has been extensively studied during the past several decades. Soliton-like solutions \cite{m, td, yang01} are of great interest among these studies. In \cite{linyang, yang00}, the authors rigorously studied the vortices and strings generated from the Born-Infeld electromagnetism. As the Euler-Lagrange equation of the Born-Infeld energy functional is a nonlinear second-order system, which is usually hard to solve, the study of these soliton-like solutions always leads to many fascinating problems in nonlinear partial differential equations. The seminal work of Taubes \cite{taubes, t2} mathematical rigorously studied the vortex of the two-dimensional Abelian Higgs model using the BPS reduction \cite{bog}, which provides an effective way to circumvent such difficulties. However, to achieve such a reduction, the model is required to have the self-dual structure, which restricts the flexibility of the Higgs potential function. Many generalized models were proposed to realize a less restrictive Higgs self-interaction pattern, such as the generalized Abelian-Higgs model \cite{gab, lohe81} and the generalized Chern-Simons-Higgs model \cite{gcs, leenam}. On the other hand, it is natural to see whether a generalized model maintains the self-dual structure. In this paper, we investigate a generalized Born-Infeld model proposed in \cite{gbf}. By introducing modifying functions into the original model, the action density of the generalized Born-Infeld model is defined to be
\begin{equation}\label{ac}
\mathcal L_{GBI}=b^2\big(1-\sqrt{1+\frac{G(|\phi|^2)}{2b^2}F_{\mu\nu}F^{\mu\nu}} \big)+w(|\phi|^2)D_\mu \phi\overline {D^{\mu} \phi} -V(|\phi|^2)
\end{equation}
where the modifying functions $G(|\phi|^2)$ and $w(|\phi|^2)$ are positive functions. In \cite{gbf}, by choosing specific $G(|\phi|^2)$ and $w(|\phi|^2)$, the authors studied the radial symmetric self-dual vortices solutions and presented several numerical results. Consequently, in \cite{han16}, with the same form of $G(|\phi|^2)$ and $w(|\phi|^2)$, the author established the existence of the self-dual vortices on a doubly periodic domain and full plane without the assumption of radial symmetry. Inspired by these results, we consider vortices and cosmic strings arising in the generalized Born-Infeld model (\ref{ac}) with a more general class of modifying functions. The Euler-Lagrange equations of  (\ref{ac}) is:
\begin{equation}\label{eq1}
\frac{1}{\sqrt{|g|}}\partial_{\mu}\Big(\frac{G(|\phi|^2)g^{\mu\mu'}g^{\alpha\nu'}\sqrt{|g|}F_{\mu'\nu'}}{\sqrt{1+\frac{G(|\phi|^2)}{2b^2}g^{\mu\mu'}g^{\nu\nu'}F_{\mu\nu}F_{\mu'\nu'}}}\Big)=j^{\alpha}
\end{equation}

\begin{equation}\label{eq2}
\frac{1}{\sqrt{|g|}}D_{\mu}\Big(w(|\phi|^2)g^{\mu\nu}\sqrt{|g|}D_{\nu}\phi\Big)=f
\end{equation}

where 
\[j^{\alpha}=\mathrm{i}w(|\phi|^2)g^{\alpha\mu}(\overline\phi D_\mu \phi-\phi\overline{D_{\mu}\phi}))\]
\[f=\Big(\frac{G'(|\phi|^2)g^{\mu\mu'}g^{\nu\nu'}F_{\mu\nu}F_{\mu'\nu'}}{4\sqrt{1+\frac{G(|\phi|^2)}{2b^2}g^{\mu\mu'}g^{\nu\nu'}F_{\mu\nu}F_{\mu'\nu'}}}+w'(|\phi|^2)D_{\mu}\phi\overline{D_{\nu}\phi}g^{\mu\nu}-V'(|\phi|^2)\Big)\phi\]
are current and force densities. We study the vortex solutions of the system (\ref{eq1}), (\ref{eq2}) under mild conditions on modifying functions $G(|\phi|^2)$ and $w(|\phi|^2)$ instead of using specific forms. Moreover, we consider the cosmic strings by coupling (\ref{eq1}), (\ref{eq2}) with the Einstein equations 
\begin{equation}\label{eq3}
G_{\mu\nu}=-8\pi GT_{\mu\nu}
\end{equation}
where the constant $G$ is Newton's gravitational constant and $G_{\mu\nu}=R_{\mu\nu}-\frac{1}{2}g_{\mu\nu}R$. Here, $R_{\mu\nu}$ and $R$ are Ricci tensor and scalar curvature, respectively.
The energy-momentum tensors $T_{\alpha\beta}$ are defined by
\begin{equation}
T_{\alpha\beta}=-\frac{G(|\phi|^2)g^{\mu'\nu'}F_{\alpha\mu'}F_{\beta\nu'}}{\sqrt{1+\frac{G(|\phi|^2)}{2b^2}g^{\mu\mu'}g^{\nu\nu'}F_{\mu\nu}F_{\mu'\nu'}}}+w(|\phi|^2)(D_{\alpha}\phi\overline{D_{\beta}\phi}+\overline{D_{\alpha}\phi}D_{\beta}\phi)-g_{\alpha\beta}\mathcal L_{GBI}
\end{equation}
where $\alpha$, $\beta=0,1,2,3$. Our main results are as follows:

(1) With new relations between $G(|\phi|^2)$, $w(|\phi|^2)$, and $V(|\phi|^2)$, we re-derive the BPS reduction for the generalized Born-Infeld model (\ref{ac}). As a result, we get the self-dual equations of vortices and cosmic strings, respectively. Such a reduction also includes several topological quantities crucial for total energy.

(2) We give assumptions on modifying functions $G(|\phi|^2)$ and $w(|\phi|^2)$ and prove the existence of self-dual vortices on the full plane under such assumptions. We give decay estimates for the solution and show that the minimal energy of (\ref{ac}) is quantized.

(3) With suitable assumptions on modifying functions and an upper bound of the number of strings, we establish the existence of cosmic strings when the underlying domain is diffeomorphic to $\mathbb R^2$ and derive the decay estimates. We also calculate the minimal energy and the total curvature. 

In the next section, we give the conditions under which the model (\ref{ac}) can achieve the self-dual structure and derive the self-dual system by BPS reduction. We then further reduce the system to an elliptic problem that governs the existence of vortices. In section 3, the existence result for vortex solutions over $\mathbb R^2$ is proved by the method of subsolutions and supersolutions. We then discuss the decay estimates for the vortex solution and show that the minimal energy of the (\ref{ac}) is quantized and controlled by the number of vortices. In section 4, we introduce the cosmic strings problem by coupling the Euler-Lagrange equations (\ref{eq1}), (\ref{eq2}) with the Einstein equations (\ref{eq3}). Similar to the vortex solution case, we reduce the system to a governing elliptic equation and focus on the case when the underlying domain is $\mathbb R^{1,1}\times S$, where $S$ is a Riemann surface. In section 5, we prove the existence of the cosmic strings when $S$ is diffeomorphic to $\mathbb R^2$. We first consider a family of regularized equations of the governing elliptic equation with a regularizing parameter and prove the existence of the solution to the regularized equation. We then show that the solution of the regularized equation converges to the solution of the original equation, which proves the existence of the cosmic strings. Decay estimates for the solution and quantized energy are also discussed. We conclude the paper by discussing some future directions and open problems.
\section{Self-Dual Vortices}
In this section, we first explore the self-dual structure in the generalized Born-Infeld model (\ref{ac}). We focus on the case when the Minkowski spacetime is of type $\mathbb R^{1,1}\times S$. Here $(S,(g_{ij}))$ is a Riemann surface with metric $(g_{ij})$. We are interested in the static field configurations $(A,\phi)$ of the system (\ref{eq1}), (\ref{eq2}) which depend only on variables $x_1$, $x_2$ on $S$, and $A_0=A_3=0$. In an isothermal coordinate, we shall assume the metric is of the form $$(g_{\mu\nu})=\mathrm{diag}\{1, -\mathrm{e}^{\eta},  -\mathrm{e}^{\eta}, -1\}$$With these assumptions, we can easily see that $F_{12}$ and $F_{21}$ are the only two nonzero terms in $F_{\mu\nu}$ and 
\begin{equation}
D_{\mu}\phi=0,\  \mu=0,3
\end{equation} 
The action density (\ref{ac}) can be simplified as
\begin{equation}
\mathcal L_{GBI}=b^2\Big(1-\sqrt{1+\frac{G(|\phi|^2)\mathrm{e}^{-2\eta}F_{12}^2}{b^2}}\Big)-\mathrm{e}^{-\eta}w(|\phi|^2)(|D_1\phi|^2+|D_2\phi|^2)-V(|\phi|^2)
\end{equation} 
and the energy density $\mathcal H=T_{00}$ can be written as
\begin{equation}
\begin{aligned}
\mathcal H&=T_{00}\notag\\
&=b^2\Big(\sqrt{1+\frac{G(|\phi|^2)\mathrm{e}^{-2\eta}F_{12}^2}{b^2}}-1\Big)+\mathrm{e}^{-\eta}w(|\phi|^2)(|D_1\phi|^2+|D_2\phi|^2)+V(|\phi|^2)\notag\\
&=b^2\Big(\sqrt{1+\frac{G(|\phi|^2)\mathrm{e}^{-2\eta}F_{12}^2}{b^2}}-1\Big)\notag\\
&\quad+\mathrm{e}^{-\eta}w(|\phi|^2)\Big(|D_1\phi\pm\mathrm{i}D_2\phi|^2 \pm\mathrm{i}(D_1\phi\overline{D_2\phi}-\overline{D_1\phi}D_2\phi)\Big)+V(|\phi|^2)
\end{aligned}
\end{equation}
the last equality is due to the identity:
\[|D_1\phi|^2+|D_2\phi|^2=|D_1\phi\pm\mathrm{i}D_2\phi|^2 \pm\mathrm{i}(D_1\phi\overline{D_2\phi}-\overline{D_1\phi}D_2\phi)\]
Let 
\begin{equation}\label{simp2}
F=\sqrt{1+\frac{G(|\phi|^2)\mathrm{e}^{-2\eta}F_{12}^2}{b^2}},\quad U=\sqrt{1-\frac{U^2(|\phi|^2)}{b^2}}
\end{equation}
then we have
\begin{equation}
\begin{aligned}
\mathcal H=&\frac{1}{2}(\mathrm{e}^{-\eta}F_{12}\sqrt{G(|\phi|^2})\pm FU(|\phi|^2)^2 F^{-1}+\frac{b^2}{2}(FU-1)^2F^{-1}\notag\\
&-b^2\mp \mathrm{e}^{-\eta}F_{12}\sqrt{G(|\phi|^2})U(|\phi|^2)+b^2U\notag\\
&+\mathrm{e}^{-\eta}w(|\phi|^2)\Big(|D_1\phi\pm\mathrm{i}D_2\phi|^2 \pm\mathrm{i}(D_1\phi\overline{D_2\phi}-\overline{D_1\phi}D_2\phi)\Big)+V(|\phi|^2)\notag
\end{aligned}
\end{equation}
To derive the self-dual reduction, we introduce the current density 
\begin{equation}
J_k=\mathrm{i}g(|\phi|^2)(\phi\overline{D_k\phi}-\overline\phi D_k\phi)
\end{equation}
then we have
\begin{equation}
\begin{aligned}
J_{12}=\partial_1J_2-\partial_2J_1&=\partial_1\Big(g(|\phi|^2)\mathrm{i}(\phi\overline{D_k\phi}-\overline\phi D_k\phi)\Big)-\partial_2(g(|\phi|^2)\mathrm{i}(\phi\overline{D_k\phi}-\overline\phi D_k\phi))\notag\\
&=2\mathrm{i}\Big(g'(|\phi|^2)|\phi|^2+g(|\phi|^2)\Big)(D_1\phi\overline{D_2\phi}-\overline{D_1\phi}D_2\phi)-2g(|\phi|^2)|\phi|^2F_{12}\notag
\end{aligned}
\end{equation}
and the energy density can be rewritten as
\begin{align}
\mathcal H&=\frac{1}{2}(\mathrm{e}^{-\eta}F_{12}\sqrt{G(|\phi|^2})\pm FU(|\phi|^2)^2 F^{-1}+\frac{b^2}{2}(FU-1)^2F^{-1}\notag\\
&\quad-b^2\mp \mathrm{e}^{-\eta}F_{12}\sqrt{G(|\phi|^2})U(|\phi|^2+b^2U\notag\\
&\quad+\mathrm{e}^{-\eta}w(|\phi|^2)|D_1\phi^2\pm\mathrm{i}D_2\phi|^2 \pm\mathrm{e}^{-\eta}w(|\phi|^2)\frac{J_{12}+2g(|\phi|^2)|\phi|^2F_{12}}{2\Big(g'(|\phi|^2)|\phi|^2+g(|\phi|^2)\Big)}+V(|\phi|^2)\notag\\
&=\frac{1}{2}(\mathrm{e}^{-\eta}F_{12}\sqrt{G(|\phi|^2)}\pm FU(|\phi|^2))^2 F^{-1}+\frac{b^2}{2}(FU-1)^2F^{-1}-b^2+b^2U+V(|\phi|^2)\notag\\
&\quad+\mathrm{e}^{-\eta}w(|\phi|^2)|D_1\phi\pm\mathrm{i}D_2\phi|^2\pm\frac{\mathrm{e}^{-\eta}w(|\phi|^2)}{2\Big(g'(|\phi|^2)|\phi|^2+g(|\phi|^2)\Big)}J_{12}\notag\\
&\quad+\Big(\pm\frac{2\mathrm{e}^{-\eta}w(|\phi|^2)g(|\phi|^2)|\phi|^2}{2\Big(g'(|\phi|^2)|\phi|^2+g(|\phi|^2)\Big)}\mp\mathrm{e}^{-\eta}\sqrt{G(|\phi|^2)}U(|\phi|^2)\Big)F_{12}\notag
\end{align}
If we assume that 
\begin{equation}\label{as0}
V(|\phi|^2)=b^2(1-U)
\end{equation}
\begin{equation}
w(|\phi|^2)=2\Big(g'(|\phi|^2)|\phi|^2+g(|\phi|^2)\Big)\label{as1}
\end{equation}
\begin{equation}
2g(|\phi|^2)|\phi|^2-\sqrt{G(|\phi|^2)}U(|\phi|^2)=1\label{as2}
\end{equation}
then $\mathcal H$ can be simplified to be:
\begin{align}\label{sph}
\mathcal H&=\frac{1}{2}(\mathrm{e}^{-\eta}F_{12}\sqrt{G(|\phi|^2})\pm FU^2(|\phi|^2))^2 F^{-1}+\frac{b^2}{2}(FU-1)^2F^{-1}\notag\\
&+\mathrm{e}^{-\eta}w(|\phi|^2)|D_1\phi\pm\mathrm{i}D_2\phi|^2\pm\mathrm{e}^{-\eta}(F_{12}+J_{12})
\end{align}
Therefore we have 
\begin{equation}
\mathcal H\ge\pm\mathrm{e}^{-\eta}(F_{12}+J_{12})
\end{equation}
and the lower bound of the total energy
 \begin{equation}
 E=\int_S\mathcal H\mathrm{d}\Omega_S
 \end{equation}
is attained if and only if all the squared terms in (\ref{sph}) vanish. Therefore, we get the following self-dual system:
\begin{equation}
\mathrm{e}^{-\eta}F_{12}\sqrt{G(|\phi|^2})\pm FU(|\phi|^2)=0\label{bps1}
\end{equation}
\begin{equation}
FU-1=0\label{bps2}
\end{equation}
\begin{equation}
D_1\phi\pm\mathrm{i}D_2\phi=0\label{bps3}
\end{equation}

The solution of the system is then an energy minimizer and satisfies the original Euler-Lagrange Equations (\ref{eq1}), (\ref{eq2}). Moreover, we can see that (\ref{bps2}) can be derived from (\ref{bps1}). Therefore, solving the system (\ref{bps1}), (\ref{bps2}) and (\ref{bps3}) is equivalent to solving equation (\ref{bps1}) and (\ref{bps3}). As a consequence, the main equations of the self-dual system can be written as:
\begin{equation}
F_{12}\pm\frac{\mathrm{e}^{\eta}U(|\phi|^2)}{\sqrt{G(|\phi|^2)}\sqrt{1-\frac{1}{b^2}U^2(|\phi|^2)}}=0\label{mainbps1}
\end{equation}

\begin{equation}
D_1\phi\pm\mathrm{i}D_2\phi=0\label{mainbps2}
\end{equation}

Next, we further reduce the self-dual systems by Taubes' trick \cite{taubes}. Without loss of generality, we choose the plus sign in the equations. We are interested in the field configurations $(A, \phi)$ with prescribed zeros of $\phi$, which represents the location of vortices. Let $v=\ln |\phi|^2$, $z=x_1+\mathrm{i}x_2$, and set $P=\{p_1, p_2, ..., p_N\}$ be the set of distinct zeros of $\phi$.  Recall that 
\[\partial=\frac{1}{2}(\partial_1-\mathrm{i}\partial_2)\quad \overline{\partial}=\frac{1}{2}(\partial_1+\mathrm{i}\partial_2)\] 
then from (\ref{mainbps2}) we have
\begin{equation}\label{fml}
\begin{aligned}
\phi(z)&=\mathrm{e}^{v(z)/2+\mathrm{i}(\sum_{s=1}^N\mathrm{arg}(z-p_s))}\\
A_1(z)&=-\Re\{2\mathrm{i}\overline\partial\ln\phi(z)\}\\
A_2(z)&=-\Im\{2\mathrm{i}\overline\partial\ln\phi(z)\}\\
\end{aligned}
\end{equation}
and away from the set $P$ we have
\[F_{12}=-\frac{1}{2}\Delta\ln|\phi|^2\]
Let $v=\ln|\phi|^2$, then the equations (\ref{mainbps1}) and (\ref{mainbps2}) can be reduced to the following single elliptic equation: 
\begin{equation}
\Delta_{g}v=\frac{2U(\mathrm{e}^v)}{\sqrt{G(\mathrm{e}^v)}\sqrt{1-\frac{1}{b^2}U^2(\mathrm{e}^v)}}+4\pi\sum_{i=1}^{N}\delta_{p_i}\label{generalellip}
\end{equation}
where $P=\{p_1,p_2,...,p_N\}$ are zeros of $\phi$, namely the vortices, and $\Delta_{g}$ is the Laplace-Beltrami operator. The solution $v$ to the equation (\ref{generalellip}) can recover the solution $\phi$ and $A_\mu$ of the original system (\ref{eq1}), (\ref{eq2}) by the formulae (\ref{fml}). We are interested in the case when $(S, g_{ij})=(\mathbb R^2, \delta_{ij})$ then the equation (\ref{generalellip}) reads as
\begin{equation}
\Delta v=\frac{2U(\mathrm{e}^v)}{\sqrt{G(\mathrm{e}^v)}\sqrt{1-\frac{1}{b^2}U^2(\mathrm{e}^v)}}+4\pi\sum_{i=1}^{N}\delta_{p_i}\label{planar1}
\end{equation}
Here $G(t)$ is the modifying function from the generalized model (\ref{ac}) and $U(t)$ is a function arises from the self-dual reduction (\ref{simp2}). From (\ref{as1}) and (\ref{as2}) we see that the functions $G(t)$ and $U(t)$ are not arbitrary. To establish the existence of vortex solutions of (\ref{planar1}), we further assume that $G(t)$ and $U(t)$ satisfy the following conditions:
\begin{assumption}\label{as3}
$U(t)$ is a $C^2$ function in $\mathbb R$ with $U(1)=0$ and $U'(t)>0$. 
\end{assumption}
\begin{assumption}\label{as4}
G(t) is a continuous positive function in $\mathbb R$ and there exist $m_G$, $M_G$ such that $0<m_G<\max_{0\le t\le1}G(t)<M_G$.
\end{assumption}
With above assumptions, we have the following existence result:
\begin{theorem}\label{thm1}
Given N prescribed points $p_1$, $p_2$,..., $p_N$ in $\mathbb R^2$, the equations (\ref{eq1}) (\ref{eq2}) have an energy-minimizing solution $(\phi, A)$ with the asymptotic limit $|\phi|=1$ at infinity so that $p_1$, $p_2$,..., $p_N$ are zeros of $\phi$ provided $G(t)$ and $U(t)$ satisfy assumption \ref{as3} and assumption \ref{as4}. Moreover, we have the decay estimates
\begin{equation}
|\phi|^2-1=O(\mathrm{e}^{-\sqrt{M}|x|}),\quad F_{12}=O(\mathrm{e}^{-\sqrt{M}|x|}),\quad |D_\mu\phi|=O(\mathrm{e}^{-\sqrt{M}|x|}),\mu=1,2
\end{equation}
where $M=\frac{2U'(1)}{\sqrt{G(1)}}$, and the minimal energy is quantized with the value: 
\begin{equation}
E=\int_{\mathbb R^2}\mathcal H=\Big|\int_{\mathbb R^2} F_{12}\Big|=2\pi N
\end{equation}
\end{theorem}

\section{Existence of Self-Dual Vortices and Decay Estimates in $\mathbb R^2$}
In this section, we prove the Theorem \ref{thm1} through a study of the governing elliptic equation derived from the previous section
\begin{equation}\label{goveq1}
\Delta v=\frac{2U(\mathrm{e}^v)}{\sqrt{G(\mathrm{e}^v)}\sqrt{1-\frac{1}{b^2}U^2(\mathrm{e}^v)}}+4\pi\sum_{i=1}^{N}\delta_{p_i}
\end{equation}
as the existence of the field configuartions $(\phi, A)$ is equivalent to the existence of the solution $v$ to the equation (\ref{goveq1}), and the boundary condition
\begin{equation}
|\phi|\rightarrow1\quad\quad as\quad |x|\rightarrow\infty\label{bdy1}
\end{equation}
implies
\begin{equation}
v\rightarrow0\quad\quad as\quad |x|\rightarrow\infty\label{bdy1}
\end{equation}
First, from maximum principle we see that $v<0$. To remove the singular terms in the equation (\ref{goveq1}), we introduce the following background functions: 
\begin{equation}\label{bk}
\begin{aligned}
u_0=\sum_{s=1}^N \ln(\frac{|x-p_s|^2}{1+|x-p_s|^2}) \\
g=\sum_{s=1}^N\frac{4}{(1+|x-p_s|^2)^2}
\end{aligned}
\end{equation}
and we have
\begin{equation}
\begin{aligned}
\Delta u_0&=4\pi\sum_{s=1}^N\delta_{p_s}-g\\
\end{aligned}
\end{equation}
Let $u=v-u_0$, then the equation (\ref{goveq1}) becomes
\begin{equation}\label{goveq11}
\Delta u=\frac{2U(\mathrm{e}^{u+u_0})}{\sqrt{G(\mathrm{e}^{u+u_0})}\sqrt{1-\frac{1}{b^2}U^2(\mathrm{e}^{u+u_0})}}+g
\end{equation}

We prove the existence of solution to (\ref{goveq11}) via the method of subsolutions and supersolutions, then we can get the solution $v$ of equation (\ref{planar1}) by $v=u+u_0$. We first find the subsolutions and supersolutions of (\ref{goveq11}) by considering a family of equations of the following type:
\begin{equation}
\Delta v_{\lambda}=\lambda(\mathrm{e}^{v_{\lambda}}-1)+4\pi\sum_{i=1}^{N}\delta_{p_i}\label{abh}
\end{equation}
\begin{equation}
v_\lambda\rightarrow0\quad\quad as\quad |x|\rightarrow\infty
\end{equation}
This type of equations has been studied extensively due to its close connection to Abelian Higgs model and prescirbed Gauss curvature problem. It is well known that the equation has a unique solution for any $\lambda>0$ \cite{taubes}. Using the background functions again, we can rewrite the equation (\ref{abh}) as
\begin{equation}
\Delta u_\lambda=\lambda(\mathrm{e}^{u_\lambda+u_0}-1)+g\label{abh2}
\end{equation}
where $u_{\lambda}=v_\lambda-u_0$. By choosing suitable value of $\lambda$, we can get the subsolutions and supersolutions of the equation (\ref{goveq11}), and we have the following existence result:
\begin{theorem}
Let $\lambda_1=\frac{2m_U}{\sqrt{M_G}}$, $\lambda_2=\frac{2M_U}{\sqrt{m_G(1-U^2(0)/b^2})}$, then $u_{\lambda_1}$ is a subsolution of the equation (\ref{goveq11}), and $u_{\lambda_2}$ is a supersolution of the equation (\ref{goveq11}). Therefore the equation (\ref{goveq11}) has a solution $v$ with $v=0$ at infinity. 
\end{theorem}
\begin{proof}
For the simplicity, we let
\[f(u,x)=\frac{2U(\mathrm{e}^{u+u_0})}{\sqrt{G(\mathrm{e}^{u+u_0})}\sqrt{1-\frac{1}{b^2}U^2(\mathrm{e}^{u+u_0})}}+g\]
By assumptions on $G(t)$ and $U(t)$, there exist $M_U$, $m_U>0$ and $M_G$, $m_G>0$, such that $m_U<U'(t)<M_U$ and $0<m_G<G(t)<M_G$ when $0<t<1$. Let $\lambda_1=\frac{2m_U}{\sqrt{M_G}}$,  then we have
\begin{align}
f(u_{\lambda_1},x)&=\frac{2U(\mathrm{e}^{u_{\lambda_1}})}{\sqrt{G(\mathrm{e}^{u_{\lambda_1}})}\sqrt{1-\frac{1}{b^2}U^2(\mathrm{e}^{u_{\lambda_1}})}}+g\notag\\
&=\frac{2U'(\xi)(\mathrm{e}^{u_{\lambda_1}}-1)}{\sqrt{G(\mathrm{e}^{u_{\lambda_1}})}\sqrt{1-\frac{1}{b^2}U^2(\mathrm{e}^{u_{\lambda_1}})}}+g\notag\\
&<\frac{2m_U(\mathrm{e}^{u_{\lambda_1}}-1)}{\sqrt{M_G}}+g\notag\\
&=\lambda_1(\mathrm{e}^{u_{\lambda_1}}-1)+g=\Delta u_{\lambda_1}\notag
\end{align}
where $\mathrm{e}^{u_{\lambda_1}}<\xi<1$. 
Therefore $u_{\lambda_1}$ is a subsolution of equation (\ref{goveq11}). To find the supersolution, let $\lambda_2=\frac{2M_U}{\sqrt{m_G(1-U^2(0)/b^2})}$ then 
\begin{align}
f(u_{\lambda_2},x)&=\frac{2U(\mathrm{e}^{u_{\lambda_2}})}{\sqrt{G(\mathrm{e}^{u_{\lambda_2}})}\sqrt{1-\frac{1}{b^2}U^2(\mathrm{e}^{u_{\lambda_2}})}}+g\notag\\
&=\frac{2U'(\theta)(\mathrm{e}^{u_{\lambda_2}}-1)}{\sqrt{G(\mathrm{e}^{u_{\lambda_2}})}\sqrt{1-\frac{1}{b^2}U^2(\mathrm{e}^{u_{\lambda_2}})}}+g\notag\\
&>\frac{2M_U(\mathrm e^{u_{\lambda_2}}-1)}{\sqrt{m_G(1-\frac{1}{b^2}U^2(0)})}+g\notag\\
&=\lambda_2(\mathrm e^{u_{\lambda_2}}-1)+g=\Delta u_{\lambda_2}\notag
\end{align}
where $\mathrm{e}^{u_{\lambda_2}}<\theta<1$. Hence $u_{\lambda_2}$ is a supersolution of equation (\ref{goveq11}). Therefore we can get the solution $u$ of equation (\ref{goveq11}) by iteration. Since $u_{\lambda_1}<u<u_{\lambda_2}$ and $u_{\lambda}\rightarrow0$ as $|x|\rightarrow\infty$, for $\lambda>0$, we have $u\rightarrow0$ as $|x|\rightarrow\infty$. \end{proof}
Let $v=u+u_0$, we get the solution to the equation (\ref{goveq1}). Next, we discuss the asymptotic behavior of the solution $v$ as $|x|\rightarrow\infty$ and show that the minimal energy is quantized. We have the following theorem:
\begin{theorem}
The solution $v$ to the equation (\ref{goveq1}) with $v=0$ at infinity satisfies exponential decay estimates:
\[|v|=O(\mathrm{e}^{-\sqrt{M}|x|})\quad |\nabla v|=O(\mathrm{e}^{-\sqrt{M}|x|})\]
where $M=\frac{2U'(1)}{\sqrt{G(1)}}$. As a consequence, we have
\[|\phi|^2-1=O(\mathrm{e}^{-\sqrt{M}|x|}),\quad F_{12}=O(\mathrm{e}^{-\sqrt{M}|x|}),\quad |D_\mu\phi|=O(\mathrm{e}^{-\sqrt{M}|x|}),\quad\mu=1,2.\]
Moreover, the corresponding field configuration $(A, \phi)$ achieve the quantized minimal energy with the value: 
\[E=\Big|\int_{\mathbb R^2} F_{12}\Big|=2\pi N\]
\end{theorem}
\begin{proof}
We consider the solution $v$ to the equation (\ref{goveq1}). We may choose $r_0>0$ large enough so that $P\subset \overline{B(r_0)}$. Let $\Omega=\mathbb R^2\setminus\overline{B(r_0)}$, we have 
\begin{equation}\label{asy0}
\Delta v=\frac{2U(\mathrm{e}^v)}{\sqrt{G(\mathrm{e}^v)}\sqrt{1-\frac{1}{b^2}U^2(\mathrm{e}^v)}}\quad x\in \Omega
\end{equation}
Since $v\rightarrow0$ as $|x|\rightarrow\infty$, linearizing the equation (\ref{asy0}), we see that there exists a function $\xi(x)$ with $\xi(x)\rightarrow0$ as $|x|\rightarrow\infty$ such that
\[\frac{2U(\mathrm{e}^v)}{\sqrt{G(\mathrm{e}^v)}\sqrt{1-\frac{1}{b^2}U^2(\mathrm{e}^v)}}=\Big[-\frac{G'(\mathrm e^{\xi})\mathrm e^{\xi}}{G^{\frac{3}{2}}(\mathrm e^{\xi})}\frac{U(\mathrm{e}^\xi)}{\sqrt{1-\frac{1}{b^2}U^2(\mathrm{e}^\xi)}}+\frac{2U'(\mathrm e^{\xi})\mathrm e^{\xi}}{\sqrt{G(\mathrm e^{\xi})}(1-U^2(\xi)/b^2)^{\frac{3}{2}}}\Big]v\]
Let 
\[F(x)=-\frac{G'(\mathrm e^{\xi})\mathrm e^{\xi}}{G^{\frac{3}{2}}(\mathrm e^{\xi})}\frac{U(\mathrm{e}^\xi)}{\sqrt{1-\frac{1}{b^2}U^2(\mathrm{e}^\xi)}}+\frac{2U'(\mathrm e^{\xi})\mathrm e^{\xi}}{\sqrt{G(\mathrm e^{\xi})}(1-U^2(\xi)/b^2)^{\frac{3}{2}}}\]
we have
\[\lim_{x\to\infty}F(x)=\frac{2U'(1)}{\sqrt{G(1)}}\]
therefore we have
\[|v|=O(\mathrm{e}^{-\sqrt{M}|x|})\] 
where $M=\frac{2U'(1)}{\sqrt{G(1)}}$. Moreover, we see that the right hand side of (\ref{asy0}) is in $L^2(\Omega)$, therefore by $L^2$-estimates we have $v\in W^{2,2}(\Omega)$. For $\partial_j v$ with fixed $j$, we have
\begin{equation}\label{asy01}
\Delta(\partial_j v)=\frac{2U'(\mathrm e^{v})\mathrm e^{v}\partial_j v}{\sqrt{G(\mathrm e^{v})}(1-U^2(v)/b^2)^{\frac{3}{2}}}-\frac{G'(\mathrm e^{v})\mathrm e^{v}}{G^{\frac{3}{2}}(\mathrm e^{v})}\frac{U(\mathrm{e}^v)\partial_j v}{\sqrt{1-\frac{1}{b^2}U^2(\mathrm{e}^v)}}
\end{equation}
the coefficient of $\partial_j v$ on the right hand side again has the limit $M$ as $|x|\rightarrow \infty$. Therefore, $|\nabla v|$ has the same decay estimates as $v$. Moreover, according to (\ref{fml}), we get the exponential decay for the Maxwell stress tensor, covariant derivatives, and the Higgs field $\phi$
\[|\phi|^2-1=O(\mathrm{e}^{-\sqrt{M}|x|}),\quad F_{12}=O(\mathrm{e}^{-\sqrt{M}|x|}),\quad |D_\mu\phi|=O(\mathrm{e}^{-\sqrt{M}|x|}),\quad\mu=1,2.\]
With the above decay estimates, we can calculate the energy $E=|\int_{\mathbb R^2}(F_{12}+J_{12})|$, we first calculate the magnetic flux term $\Phi=\int_{\mathbb R^2}F_{12}$, from equation (\ref{mainbps1}), we see that
\begin{equation}
\int_{\mathbb R^2}F_{12}=-\int_{\mathbb R^2}\frac{U(|\phi|^2)}{\sqrt{G(|\phi|^2)}\sqrt{1-\frac{1}{b^2}U^2(|\phi|^2)}}=-\int_{\mathbb R^2}\frac{U(\mathrm e^v)}{\sqrt{G(\mathrm e^v)}\sqrt{1-\frac{1}{b^2}U^2(\mathrm e^v)}}
\end{equation}
Substituting $v$ with $u+u_0$, from equation (\ref{goveq11}), we have
\begin{equation}
\begin{aligned}
\int_{\mathbb R^2}F_{12}&=-\int_{\mathbb R^2}\frac{U(\mathrm{e}^{u+u_0})}{\sqrt{G(\mathrm{e}^{u+u_0})}\sqrt{1-\frac{1}{b^2}U^2(\mathrm{e}^{u+u_0})}}\\
&=-\frac{1}{2}\int_{\mathbb R^2}\Delta u+\frac{1}{2}\int_{\mathbb R^2}g\\
&=-\frac{1}{2}\lim_{r\to\infty}\oint_{|x|=r}\frac{\partial u}{\partial n}\mathrm d s+\frac{1}{2}\int_{\mathbb R^2}g
\end{aligned}
\end{equation}
From the decay estimates for $|\nabla v|$ and $|\nabla u|$, we see that line integrals 
\[\lim_{r\to\infty}\oint_{|x|=r}\frac{\partial u}{\partial n}\mathrm d s=0\]
Therefore
\begin{equation}
\Phi=\int_{\mathbb R^2}F_{12}=\frac{1}{2}\int_{\mathbb R^2}g=2\pi N
\end{equation}
On the other hand,
\begin{equation}
\int_{\mathbb R^2}J_{12}=\lim_{r\to\infty}\oint_{|x|=r}J_k\mathrm d x^k=\lim_{r\to\infty}\oint_{|x|=r}g(\mathrm e^v)\mathrm e^v(-\partial_2v\mathrm d x^1+\partial_1 v\mathrm d x^2)=0
\end{equation}
due to the decay estimates for $|\nabla v|$. Therefore we get the quantized energy
\begin{equation}
E=\int_{\mathbb R^2}\mathcal H=2\pi N
\end{equation}
\end{proof}

\section{Self-Dual Cosmic Strings}
In this section, we study the cosmic string solutions arising in the generalized Born-Infeld model by considering the coupled system (\ref{eq1}), (\ref{eq2}), (\ref{eq3}). Similar to the vortex solutions case, we consider the problem in the space of type $\mathbb R^{1,1}\times S$, where $S$ is noncompact and complete. We are interested in the static solution $(A, \phi)$ such that $A=(0, A_1, A_2, 0)$ and $\phi$ depends only on variables $x_1$ and $x_2$. From the previous discussion, we see that with the assumptions \ref{as0}, \ref{as1}, and \ref{as2}, the equations (\ref{eq1}) and (\ref{eq2}) can be reduced to 
\begin{equation}
F_{12}+\frac{\mathrm{e}^{\eta}U(|\phi|^2)}{\sqrt{G(|\phi|^2)}\sqrt{1-\frac{1}{b^2}U^2(|\phi|^2)}}=0\label{cos1}
\end{equation}

\begin{equation}
D_1\phi+\mathrm{i}D_2\phi=0\label{cos2}
\end{equation}
Next we show that the Einstein equations (\ref{eq3}) can also be simplified. We consider the energy-momentum tensors $T_{\alpha\beta}$.
\[\mathcal H=T_{00}=b^2\Big(\sqrt{1+\frac{G(|\phi|^2)\mathrm{e}^{-2\eta}F_{12}^2}{b^2}}-1\Big)+\mathrm{e}^{-\eta}w(|\phi|^2)(|D_1\phi|^2+|D_2\phi|^2)+V(|\phi|^2)=-T_{33}\notag
\]
Since $F_{0\mu'}=0$ for all $\mu'$, and $g_{0\mu}=0$ when $\mu\ne0$, we have 
\[T_{0\mu}=\frac{-g^{\mu'\nu'}F_{0\mu'}F_{\mu\nu'}G(|\phi|^2)}{\sqrt{1+\frac{G(|\phi|^2)\mathrm{e}^{-2\eta}F_{12}^2}{b^2}}}-g_{0\mu}L=0  \quad when\ \mu\ne0\]
Similarly, we have
\[T_{3\mu}=0 \quad when \quad \mu=0,1,2\]
Moreover, from self-dual reduction, we have
\[
T_{12}=T_{21}=w(|\phi|^2)\Big(D_1\phi\overline{D_2\phi}+\overline{D_1\phi}D_2\phi\Big)=0
\]
\begin{align}
T_{11}&=\frac{-g^{\mu'\nu'}F_{1\mu'}F_{1\nu'}G(|\phi|^2)}{\sqrt{1+\frac{G(|\phi|^2)\mathrm{e}^{-2\eta}F_{12}^2}{b^2}}}+2w(|\phi|^2)|D_1\phi|^2\notag\\
&+\mathrm{e}^{\eta}\Big(b^2(1-\sqrt{1+\frac{G(|\phi|^2)\mathrm{e}^{-2\eta}F_{12}^2}{b^2}})-\mathrm{e}^{-\eta}w(|\phi|^2)(|D_1\phi|^2+|D_2\phi|^2)-V(|\phi|^2)\Big)\notag\\
&=\frac{\mathrm{e}^{\eta}(F^2-1)b^2}{F}+w(|\phi|^2)(|D_1\phi|^2-|D_2\phi|^2)+\mathrm{e}^{\eta}(b^2(1-F)-b^2(1-U))=0\notag
\end{align}
\begin{align}
T_{22}&=\frac{-g^{\mu'\nu'}F_{2\mu'}F_{2\nu'}G(|\phi|^2)}{\sqrt{1+\frac{G(|\phi|^2)\mathrm{e}^{-2\eta}F_{12}^2}{b^2}}}+2w(|\phi|^2)|D_2\phi|^2\notag\\
&+\mathrm{e}^{\eta}\Big(b^2(1-\sqrt{1+\frac{G(|\phi|^2)\mathrm{e}^{-2\eta}F_{12}^2}{b^2}})-\mathrm{e}^{-\eta}w(|\phi|^2)(|D_1\phi|^2+|D_2\phi|^2)-V(|\phi|^2)\Big)\notag\\
&=\frac{\mathrm{e}^{\eta}(F^2-1)b^2}{F}+w(|\phi|^2)(|D_2\phi|^2-|D_1\phi|^2)+\mathrm{e}^{\eta}(b^2(1-F)-b^2(1-U))=0\notag
\end{align}
Therefore, there are only two nonzero energy-momentum tensors, which are $T_{00}$ and $T_{33}$. Moreover, with the metric $(g_{\mu\nu})=\mathrm{diag}\{1, -\mathrm{e}^{\eta},  -\mathrm{e}^{\eta}, -1\}$, Einstein tensors are 
\[G_{00}=-G_{33}=-K_g, \quad G_{\mu\nu}=0,\quad \mathrm{otherwise}\]
where $K_g$ is the Gaussian curvature of $S$. So the Einstein equations (\ref{eq3}) can be reduced to
\begin{equation}\label{cos3}
\begin{aligned}
K_g=8\pi G\mathcal H\\
T_{ij}=0 \quad  i,j=1,2
\end{aligned}
\end{equation}
Combining this result with self-dual reduction, we derive the reduced coupled system (\ref{cos1}), (\ref{cos2}), and (\ref{cos3}). When $S$ is noncompact and complete, we may assume that $(S, g_{ij})=({\mathbb R^2, \mathrm e^{\eta}\delta_{ij}})$. The solution $(\phi, A, g)$ to the system with prescribed $N$ distinct zeros of $\phi$ represents the cosmic strings. Similar to the vortex solution case, we assume that the modifying functions $G(t)$, $U(t)$, and $w(t)$ satisfy the following conditions:
\begin{assumption}\label{as21}
$U(t)$ is a $C^2$ function in $\mathbb R$ with $U(1)=0$ and $U'(t)>0$. 
\end{assumption}
\begin{assumption}\label{as22}
G(t) is a continuous positive function in $\mathbb R$ and there exist $m_G$, $M_G$ such that $0<m_G<\max_{0\le t\le1}G(t)<M_G$.
\end{assumption}
\begin{assumption}\label{as23}
$g(t)>0$ when $t\in(0,1)$
\end{assumption}
With above assumptions, we shall establish the existence theorem for cosmic string solutions:
\begin{theorem}\label{costh}
Given $N$ points $p_1$, $p_2$,...$p_N$ in $\mathbb R^2$ satisfying $8\pi GN<1$, the coupled system (\ref{cos1}), (\ref{cos2}), and (\ref{cos3}) has a energy-minimizing solution $(\phi, A, g)$ so that $p_1$, $p_2$,...$p_N$ are zeros of $\phi$, provided assumptions (\ref{as21}), (\ref{as22}) and (\ref{as23}) are satisfied. Moreover, the conformal factor $\mathrm e^{\eta}$ satisfies the decay estimate 
\[\mathrm e^{\eta(x)}=O(|x|^{-16\pi GN})\quad as\quad |x| \to\infty\]
and for any $b>0$, 
\[|\phi|^2-1=O(|x|^{-b}),\quad F_{12}=O(|x|^{-b}),\quad |D_\mu\phi|=O(|x|^{-b}),\quad\mu=1,2\]
as $|x|\to\infty$ and we have the quantized result for magnetic flux, minimal energy, and total curvature
\[\Phi=\int_{\mathbb R^2}F_{12}=2\pi N,\quad E=\int_{\mathbb R^2}\mathcal H\mathrm e^{\eta}=2\pi N,\quad  \int_{\mathbb R^2}K_g\mathrm d \Omega_g=16\pi^2GN\]
\end{theorem}
It is interesting to compare this result to the vortex solution case, when there is no effect of gravity. Here we see that the existence of gravity induces an upper bound for the number of cosmic strings and changes the decay rate of the solution. We will give the proof of the theorem in the next section. 

\section{Existence of Self-Dual Cosmic Strings in $\mathbb R^2$ and Decay Estimates}
In this section we prove the existence of the cosmic strings solution to the system (\ref{cos1}), (\ref{cos2}), and (\ref{cos3}). We first solve the Einstein equations (\ref{cos3}) and find the conformal factor $\mathrm e^{\eta}$. Again, we use the substitution $v=\ln|\phi|^2$, from (\ref{fml}), we have $|D_1\phi|^2+|D_2\phi|^2=\frac{1}{2}\mathrm{e}^v|\nabla v|^2$. Plugging this into $\mathcal H$, we get
\begin{align}
\mathcal H&=b^2(F-1)+\mathrm{e}^{-\eta}w(|\phi|^2)(|D_1\phi|^2+|D_2\phi|^2)+b^2(1-U)\notag\\
&=\frac{\mathrm{e}^{-\eta}}{2}w(\mathrm{e}^v)\mathrm{e}^v|\nabla v|^2+b^2(F-U)\notag\\
&=\frac{\mathrm{e}^{-\eta}}{2}w(\mathrm{e}^v)\mathrm{e}^v|\nabla v|^2+\frac{U^2(\mathrm{e}^v)}{\sqrt{1-\frac{1}{b^2}U^2(\mathrm{e}^v)}}\notag
\end{align}
Let $F(t)$ be a smooth function such that $F'(t)=2g(t)$, from assumption (\ref{as1}), (\ref{as2}) and equation (\ref{cos1}), we have
\begin{align}
\mathcal H&=\frac{\mathrm{e}^{-\eta}}{2}w(\mathrm{e}^v)\mathrm{e}^v|\nabla v|^2+\frac{U^2(\mathrm{e}^v)}{\sqrt{1-\frac{1}{b^2}U^2(\mathrm{e}^v)}}\notag\\
&=\frac{\mathrm{e}^{-\eta}}{2}\Big(\Delta F(\mathrm{e}^v)-F'(\mathrm{e}^v)\mathrm{e}^v\Delta v\Big)+\frac{U^2(\mathrm{e}^v)}{\sqrt{1-\frac{1}{b^2}U^2(\mathrm{e}^v)}}\notag\\
&=\frac{1}{2}\Delta_{g}F(\mathrm{e}^v)+F'(\mathrm{e}^v)\mathrm{e}^v\mathrm{e}^{-\eta}F_{12}+\frac{U^2(\mathrm{e}^v)}{\sqrt{1-\frac{1}{b^2}U^2(\mathrm{e}^v)}}\notag\\
&=\frac{1}{2}\Delta_{g}F(\mathrm{e}^v)-F'(\mathrm{e}^v)\mathrm{e}^v\frac{U(\mathrm{e}^v)}{\sqrt{1-\frac{1}{b^2}U^2(\mathrm{e}^v)}\sqrt{G(\mathrm{e}^v)}}+\frac{U^2(\mathrm{e}^v)}{\sqrt{1-\frac{1}{b^2}U^2(\mathrm{e}^v)}}\notag\\
&=\frac{1}{2}\Delta_{g}F(\mathrm{e}^v)+\frac{U(\mathrm{e}^v)\Big(-F'(\mathrm{e}^v)\mathrm{e}^v+U(\mathrm{e}^v)\sqrt{G(\mathrm{e}^v)}\Big)}{\sqrt{G(\mathrm{e}^v)}\sqrt{1-\frac{1}{b^2}U^2(\mathrm{e}^v)}}\notag\\
&=\frac{1}{2}\Delta_{g}F(\mathrm{e}^v)-\frac{U(\mathrm{e}^v)}{\sqrt{G(\mathrm{e}^v)}\sqrt{1-\frac{1}{b^2}U^2(\mathrm{e}^v)}}\notag
\end{align}
we can then simplify the Einstein equations 
\begin{equation}\label{spcos3}
K_g=4\pi G\Delta_{g}F(\mathrm{e}^v)-\frac{8\pi GU(\mathrm{e}^v)}{\sqrt{G(\mathrm{e}^v)}\sqrt{1-\frac{1}{b^2}U^2(\mathrm{e}^{v})}}
\end{equation}
Recall that for a metric $g$ which is conformal to a known metric $g_0$, we have $g=\mathrm{e}^{\eta}g_0$ and $\Delta_{g}=\mathrm{e}^{-\eta}\Delta_{g_0}$ for some function $\eta$, and the corresponding Gauss curvatures $K_g$ and $K_{g_0}$ satisfy the equation
\begin{equation}
-\Delta_{g_0}\eta+2K_{g_0}=2K_g\mathrm{e}^{\eta}
\end{equation}
Combining this with equation (\ref{spcos3}), we get
\begin{equation}
\Delta_{g_0}(\eta+8\pi GF(\mathrm{e}^v))=2K_{g_0}+16\pi G\mathrm{e}^{\eta}\frac{U(\mathrm{e}^v)}{\sqrt{G(\mathrm{e}^v)}\sqrt{1-\frac{1}{b^2}U^2(\mathrm{e}^v)}}
\end{equation}
Since we are interested the flat metric $g_0$ on $\mathbb R^2$, we have $K_{g_0}=0$, $\Delta_{g_0}=\Delta$, which is the standard Laplace operator on $\mathbb R^2$. The above equation is finally reduced to:
\begin{equation}\label{cos12}
\Delta(\eta+8\pi GF(\mathrm{e}^v))=16\pi G\mathrm{e}^{\eta}\frac{U(\mathrm{e}^v)}{\sqrt{G(\mathrm{e}^v)}\sqrt{1-\frac{1}{b^2}U^2(\mathrm{e}^v)}}
\end{equation}
On the other hand, from the previous discussion, we see that equations (\ref{cos1}) and (\ref{cos2}) can be reduced to the elliptic equation 
\begin{equation}\label{cos22}
\Delta v=\frac{2\mathrm{e}^{\eta}U(\mathrm{e}^v)}{\sqrt{G(\mathrm{e}^v)}\sqrt{1-\frac{1}{b^2}U^2(\mathrm{e}^v)}}+4\pi\sum_{i=1}^{N}\delta_{p_i}
\end{equation}
Therefore we finally reduce the system (\ref{cos1}), (\ref{cos2}), and (\ref{cos3}) to the coupled system (\ref{cos12}) and (\ref{cos22}).
Moreover, we have
\begin{equation}
\Delta(\eta+8\pi GF(\mathrm e^v)-8\pi Gv)+32G\pi^2\sum_{s=1}^N\delta_{p_s}=0
\end{equation}
Therefore we see that the function 
\[h=\frac{\eta}{16\pi G}+\frac{1}{2}F(\mathrm e^v)-\frac{1}{2}v+ \sum_{s=1}^N\ln|x-p_s|\]
is an entire harmonic function, and we are interested in the case when $h$ is a constant $c$, then we have
\begin{equation}
\eta=16\pi Gc-8\pi GF(\mathrm e^v)+8\pi Gv-16\pi G\sum_{s=1}^N\ln|x-p_s|
\end{equation}
and
\begin{equation}
\mathrm e^{\eta}=\mathrm e^{16\pi Gc-8\pi GF(\mathrm e^v)+8\pi Gv}\Big(\prod_{s=1}^N |x-p_s|^{-2}\Big)^{8\pi G}
\end{equation}
which resolves the equation (\ref{cos12}).
We can now rewrite the equation (\ref{cos22}) to be 
\begin{equation}\label{gove2}
\Delta v=\frac{2\mathrm e^{16\pi Gc-8\pi GF(\mathrm e^v)+8\pi Gv}U(\mathrm{e}^v)}{\sqrt{G(\mathrm{e}^v)}\sqrt{1-\frac{1}{b^2}U^2(\mathrm{e}^v)}}\Big(\prod_{s=1}^N |x-p_s|^{-2}\Big)^{8\pi G}+4\pi\sum_{i=1}^{N}\delta_{p_i}
\end{equation}
The existence part of Theorem \ref{costh} is then equivalent to the following theorem:
\begin{theorem}\label{cose}
The equation (\ref{gove2}) has a solution with boundary value $v=0$ at infinity if $8\pi GN<1$ and assumptions \ref{as21}, \ref{as22} and \ref{as23} are satisfied.
\end{theorem}
\subsection{Proof of Existence}
In this subsection, we give the proof of the Theorem \ref{cose}. Using the background functions (\ref{bk}), and let $u=v-u_0$, we can rewrite the equation (\ref{gove2}) to be
\begin{equation}\label{gove22}
\Delta u=\frac{2\mathrm e^{16\pi Gc-8\pi GF(\mathrm e^{u+u_0})+8\pi Gu}U(\mathrm e^{u+u_0})}{\sqrt{G(\mathrm{e}^{u+u_0}})\sqrt{1-\frac{1}{b^2}U^2(\mathrm{e}^{u+u_0})}}\prod_{s=1}^N(\frac{1}{1+|x-p_s|^2})^{8\pi G}+g
\end{equation}
We see that the existence of solution to equation (\ref{gove2}) is equivalent to equation (\ref{gove22}) with the same boundary condition. To establish the the existence result, we first consider the following regularized equation
\begin{equation}\label{rgove}
\Delta u=\frac{2\mathrm e^{8\pi G({2c-F(\mathrm e^{u+u_{0,\delta}})+u})}U(\mathrm e^{u+u_{0,\delta}})}{\sqrt{G(\mathrm{e}^{u+u_{0,\delta}}})\sqrt{1-\frac{1}{b^2}U^2(\mathrm{e}^{u+u_{0,\delta}})}}\prod_{s=1}^N(\frac{1}{1+|x-p_s|^2})^{8\pi G}+g
\end{equation}
where \[u_{0,\delta}=\sum_{s=1}^N\ln\frac{\delta+|x-p_s|^2}{1+|x-p_s|^2}\] and 
\[0\le\delta<1\]for the convenience we let 
\[f(u,x)=\frac{2\mathrm e^{8\pi G({2c-F(\mathrm e^{u+u_{0,\delta}})+u})}U(\mathrm e^{u+u_{0,\delta}})}{\sqrt{G(\mathrm{e}^{u+u_{0,\delta}}})\sqrt{1-\frac{1}{b^2}U^2(\mathrm{e}^{u+u_{0,\delta}})}}\prod_{s=1}^N(\frac{1}{1+|x-p_s|^2})^{8\pi G}+g\]
\begin{lemma}\label{lm1}
$-u_{0,\delta}$ is a supersolution of equation (\ref{rgove})
\end{lemma}
\begin{proof}
Since $U(\mathrm e^{-u_{0,\delta}+u_{0,\delta}})=U(\mathrm e^0)=U(1)=0$, we have $f(-u_{0,\delta},x)=g$. On the other hand, 
\[\Delta(-u_{0,\delta})=g-\sum_{s=1}^N\frac{4\delta}{(\delta+|x-p_s|^2)^2}<f(-u_{0,\delta},x)\]
Therefore $-u_{0,\delta}$ is a supersolution of equation $(4.6)$.
\end{proof}
\begin{lemma}\label{lm2}
If $8\pi G<1$, there exists $c_0>0$ such that when $c>c_0$, $u=0$ is a subsolution of equation (\ref{rgove}) for $0<\delta<\frac{1}{2}$.
\end{lemma}
\begin{proof}
Since $\Delta 0=0$, in order to show that $0$ is a subsolution, we need to show that $0>f(0,x)$.
\begin{align}
f(0,x)&=\prod_{s=1}^N(\frac{1}{1+|x-p_s|^2})^{8\pi G}\frac{2\mathrm e^{8\pi G({2c-F(\mathrm e^{u_{0,\delta}})})}U(\mathrm e^{u_{0,\delta}})}{\sqrt{G(\mathrm{e}^{u_{0,\delta}}})\sqrt{1-\frac{1}{b^2}U^2(\mathrm{e}^{u_{0,\delta}})}}+g\notag\\
&=2\mathrm e^{16\pi Gc}\prod_{s=1}^N(\frac{1}{1+|x-p_s|^2})^{8\pi G}\frac{\mathrm e^{{-8\pi GF(\mathrm e^{u_{0,\delta}}})}U(\mathrm e^{u_{0,\delta}})}{\sqrt{G(\mathrm{e}^{u_{0,\delta}}})\sqrt{1-\frac{1}{b^2}U^2(\mathrm{e}^{u_{0,\delta}})}}+g\notag
\end{align}
Since $u_{0,\delta}<0$, we have $\frac{U(\mathrm e^{u_{0,\delta}})}{\sqrt{1-\frac{1}{b^2}U^2(\mathrm{e}^{u_{0,\delta}})}}<U(\mathrm e^{u_{0,\delta}})<0$, therefore 
\[f(0,x)<2\pi\mathrm e^{16\pi Gc}\prod_{s=1}^N(\frac{1}{1+|x-p_s|^2})^{8\pi G}\Big(\frac{\mathrm e^{{-8\pi GF(\mathrm e^{u_{0,\delta}}})}U(\mathrm e^{u_{0,\delta}})}{\sqrt{G(\mathrm{e}^{u_{0,\delta}}})}\Big)+g\]
Since $U(1)=0$, we have 
\[U(\mathrm e^{u_{0,\delta}})=U(\mathrm e^{u_{0,\delta}})-U(\mathrm e^0)=U'(\mathrm e^\xi)\mathrm e^\xi u_{0,\delta}\] 
where $u_{0,\delta} <\xi<0$. Moreover, by the assumption \ref{as23}, we have $F(\mathrm e^{u_{0,\delta}})<F(\mathrm e^0)=F(1)$, therefore 
\[\mathrm e^{-8\pi GF(u_{0,\delta})} U(u_{0,\delta})<\mathrm e^{-8\pi GF(1)}U(u_{0,\delta})=\mathrm e^{-8\pi GF(1)}U'(\mathrm e^\xi)\mathrm e^\xi u_{0,\delta} \] and
\begin{align}
2&\pi\mathrm e^{16\pi Gc}\prod_{s=1}^N(\frac{1}{1+|x-p_s|^2})^{8\pi G}\Big(\frac{\mathrm e^{{-8\pi GF(\mathrm e^{u_{0,\delta}}})}U(\mathrm e^{u_{0,\delta}})}{\sqrt{G(\mathrm{e}^{u_{0,\delta}}})}\Big)+g\notag\\
&<2\pi\mathrm e^{16\pi Gc}\prod_{s=1}^N(\frac{1}{1+|x-p_s|^2})^{8\pi G}\Big(\frac{\mathrm e^{-8\pi GF(1)}U'(\mathrm e^\xi)\mathrm e^\xi u_{0,\delta}}{\sqrt{G(\mathrm{e}^{u_{0,\delta}}})}\Big)+g\notag
\end{align}
On the other hand, 
\begin{align}
u_{0,\delta}&=\sum_{s=1}^N\ln\frac{\delta+|x-p_s|^2}{1+|x-p_s|^2}\notag\\
&=-\sum_{s=1}^N\ln\frac{1+|x-p_s|^2}{\delta+|x-p_s|^2}\notag\\
&=-\sum_{s=1}^N\frac{1}{1+\zeta_s}\frac{1-\delta}{\delta+|x-p_s|^2}\notag
\end{align}
where $0<\zeta_s<\frac{1-\delta}{\delta+|x-p_s|^2}$, $s=1,2...N$. Assuming that $0<\delta<\frac{1}{2}$, then $u_{0,\delta}\to0$ uniformly as $|x|\to\infty$.
Therefore we have,
\begin{align}
2&\pi\mathrm e^{16\pi Gc}\prod_{s=1}^N(\frac{1}{1+|x-p_s|^2})^{8\pi G}\Big(\frac{\mathrm e^{-8\pi GF(1)}U'(\mathrm e^\xi)\mathrm e^\xi u_{0,\delta}}{\sqrt{G(\mathrm{e}^{u_{0,\delta}}})}\Big)+g\notag\\
=-2&\pi\mathrm e^{16\pi Gc}\prod_{s=1}^N(\frac{1}{1+|x-p_s|^2})^{8\pi G}\Big(\frac{\mathrm e^{-8\pi GF(1)}U'(\mathrm e^\xi)\mathrm e^\xi}{\sqrt{G(\mathrm{e}^{u_{0,\delta}}})}\Big)\sum_{s=1}^N\frac{1}{1+\zeta_s}\frac{1-\delta}{\delta+|x-p_s|^2}+g\notag
\end{align}
Let \[\tilde f(u,x)=-2\pi\mathrm e^{16\pi Gc}\prod_{s=1}^N(\frac{1}{1+|x-p_s|^2})^{8\pi G}\Big(\frac{\mathrm e^{-8\pi GF(1)}U'(\mathrm e^\xi)\mathrm e^\xi}{\sqrt{G(\mathrm{e}^{u_{0,\delta}}})}\Big)\sum_{s=1}^N\frac{1}{1+\zeta_s}\frac{1-\delta}{\delta+|x-p_s|^2}\]
since $G(t)$ is positive and bounded in $[0,1]$, asymptotically we have
\[|x|^4\tilde f(u,x)\to\infty \] 
uniformly as $|x|\to\infty$ if $16\pi GN<2$.
Therefore there exists $r_0>0$ and $c_0>0$, large enough such that
\begin{align}
f(0,x)&<2\pi\mathrm e^{16\pi Gc}\prod_{s=1}^N(\frac{1}{1+|x-p_s|^2})^{8\pi G}\Big(\frac{\mathrm e^{{-8\pi GF(\mathrm e^{u_{0,\delta}}})}U(\mathrm e^{u_{0,\delta}})}{\sqrt{G(\mathrm{e}^{u_{0,\delta}}})}\Big)+g\notag\\
&<-2\pi\mathrm e^{16\pi Gc}\prod_{s=1}^N(\frac{1}{1+|x-p_s|^2})^{8\pi G}\Big(\frac{\mathrm e^{-8\pi GF(1)}U'(\mathrm e^\xi)\mathrm e^\xi}{\sqrt{G(\mathrm{e}^{u_{0,\delta}}})}\Big)\sum_{s=1}^N\frac{1}{1+\zeta_s}\frac{1-\delta}{\delta+|x-p_s|^2}+g\notag\\
&=\frac{1}{|x|^4}(-2\pi\mathrm e^{16\pi Gc}|x|^4\tilde f(u,x)+|x|^4g)<0\notag
\end{align}
when $|x|>r_0$ and $c>c_0$.

When $|x|<r_0$, $0<\delta<\frac{1}{2}$, so we can choose $c$ large enough such that 
\begin{align}
f(0,x)&<2\pi\mathrm e^{16\pi Gc}\prod_{s=1}^N(\frac{1}{1+|x-p_s|^2})^{8\pi G}\Big(\frac{\mathrm e^{{-8\pi GF(\mathrm e^{u_{0,\delta}}})}U(\mathrm e^{u_{0,\delta}})}{\sqrt{G(\mathrm{e}^{u_{0,\delta}}})}\Big)+g\notag\\
&<2\pi\mathrm e^{16\pi Gc}\prod_{s=1}^N(\frac{1}{1+|x-p_s|^2})^{8\pi G}\Big(\frac{\mathrm e^{{-8\pi GF(1)}}U(\mathrm e^{u_{0, 1/2}})}{\sqrt{M_G}}\Big)+g\notag
\end{align}
We can choose $c>c_0$ large enough so that $f(0,x)<0$ when $|x|<r_0$. Therefore $u=0$ is a subsolution of equation (\ref{rgove}) for $0<\delta<\frac{1}{2}$.
\end{proof}

From lemmas (\ref{lm1}) and (\ref{lm2}) we see that the equation (\ref{rgove}) has a solution $u_{\delta}$ for $0<\delta<\frac{1}{2}$. Now let $\{u_{\delta}\}$ be a family of solutions of the equation, we would like to pass the limit $\delta\to0$ and show that the limiting function exists and is exactly the solution of equation (\ref{gove22}). 

\begin{proof}[Proof of Theorem \ref{cose}]
Let $u_{\delta}$ be the solution of the equation (\ref{rgove}), we have 
\[ \Delta u_{\delta}=f(u_{\delta},x)\] 
recall that $f(u,x)$ is defined to be 
\[f(u,x)=\prod_{s=1}^N(\frac{1}{1+|x-p_s|^2})^{8\pi G}\frac{2\mathrm e^{8\pi G({2c-F(\mathrm e^{u+u_{0,\delta}})+u})}U(\mathrm e^{u+u_{0,\delta}})}{\sqrt{G(\mathrm{e}^{u+u_{0,\delta}}})\sqrt{1-\frac{1}{b^2}U^2(\mathrm{e}^{u+u_{0,\delta}})}}+g\]
Since $-u_{0,\delta}$ is a supersolution and $0$ is a subsolution, we have 
\[-u_0>-u_{0,\delta}\ge u_{\delta}\ge0\quad\quad \forall x\in\mathbb R^2\]
We now consider
\begin{align}
f(u_{\delta},x)&=\prod_{s=1}^N(\frac{1}{1+|x-p_s|^2})^{8\pi G}\frac{2\mathrm e^{8\pi G({2c-F(\mathrm e^{u_\delta+u_{0,\delta}})+u_\delta})}U(\mathrm e^{u_\delta+u_{0,\delta}})}{\sqrt{G(\mathrm{e}^{u_\delta+u_{0,\delta}}})\sqrt{1-\frac{1}{b^2}U^2(\mathrm{e}^{u_\delta+u_{0,\delta}})}}+g\notag\\
&=\mathrm e^{8\pi Gu_\delta}\prod_{s=1}^N(\frac{1}{1+|x-p_s|^2})^{8\pi G}\frac{2\mathrm e^{8\pi G({2c-F(\mathrm e^{u_\delta+u_{0,\delta}})})}U(\mathrm e^{u_\delta+u_{0,\delta}})}{\sqrt{G(\mathrm{e}^{u_\delta+u_{0,\delta}}})\sqrt{1-\frac{1}{b^2}U^2(\mathrm{e}^{u_\delta+u_{0,\delta}})}}+g\notag
\end{align}
Since $u_\delta+u_{0,\delta}<0$, from assumption \ref{as21}, \ref{as22}, and \ref{as23} we see that the term
\[\frac{2\mathrm e^{8\pi G({2c-F(\mathrm e^{u_\delta+u_{0,\delta}})})}U(\mathrm e^{u_\delta+u_{0,\delta}})}{\sqrt{G(\mathrm{e}^{u_\delta+u_{0,\delta}}})\sqrt{1-\frac{1}{b^2}U^2(\mathrm{e}^{u_\delta+u_{0,\delta}})}}\]
is bounded. Moreover
\begin{equation}
\prod_{s=1}^N(\frac{1}{1+|x-p_s|^2})^{8\pi G}\mathrm{e}^{8\pi Gu_{\delta}}\notag<\prod_{s=1}^N(\frac{1}{1+|x-p_s|^2})^{8\pi G}\mathrm{e}^{-8\pi Gu_0}<\prod_{s=1}^N(\frac{1}{|x-p_s|^2})^{8\pi G}
\end{equation}
When $8\pi GN<1$, there exists $p>1$ such that $\prod_{s=1}^N(\frac{1}{|x-p_s|^2})^{8\pi G}\in L_{loc}^p(\mathbb R^2)$, therefore for such $p>1$ and any bounded domain $\mathcal O\subseteq \mathbb R^2$ we have $||f(u_\delta, x)||_{L^p(\mathcal O)}\le C_1(p, \mathcal O)$ for some constant $C_1(p, \mathcal O)$ independent of $\delta$. We can then show that $u_\delta$ is bounded in $C^{2,\alpha}(\bar{\mathcal O})$ by a bootstrap argument. By $L^p$-estimates we have $||u_\delta||_{W^{2,p}(\mathcal O)}\le C_2(p, \mathcal O)$. From Sobolev Embedding we have $u_\delta$ is bounded in $C(\bar{\mathcal O})$. And threrfore $u_\delta$ is uniformly bounded in $\mathbb R^2$ as $0\le u_{\delta}<-u_0$. So $f(u_\delta, x)$ is also bounded in $\mathbb R^2$. By elliptic interior estimates, for every $p>1$ and bounded set $\mathcal O$, there exists $M$ such that $||u_\delta||_{W^{2,p}(\mathcal O)}\le M$. Take $p>2$ and again by Sobolev Embedding we have $u_\delta$ is uniformly bounded in $C^1(\bar{\mathcal O})$. Then $f(u_\delta,x)$ is also bounded in $C^1(\bar{\mathcal O})$. Therefore, by Schauder estimates we get $u_\delta$ is bounded in $C^{2,\alpha}(\bar{\mathcal O})$.

We now show the existence of solution to equation (\ref{gove22}) by a diagonal process. Let $\{\delta_{ij}\}$ be a sequence such that $\delta_{ij}\rightarrow 0$ as $j\rightarrow \infty$. We consider a family of balls $B_i=\{x\in\mathbb R^2: |x|\le i\}$. Since the embedding $C^{2,\alpha}(B_1)\rightarrow C^2(B_1)$ is compact, there exists a sequence $\{u_{\delta_{1j}}\}^{\infty}_{j=1}$ which converges in $C^2(B_1)$. Let $\lim_{j\to\infty} u_{\delta_{1j}}=u_1(x)$, $x\in B_1$. Repeat this procedure, we can extract a subsequence $\{\delta_{2j}\}\subset\{\delta_{1j}\}$ such that $u_{\delta_{2j}}$ converges to $u_2$ in $C^2(B_2)$ and $u_2|_{B_1}=u_1$. By repeating this procedure we can further get a subsequence $\{u_{\delta_{ij}}\}^{\infty}_{j=1}\subset\{u_{\delta_{(i-1)j}}\}^{\infty}_{j=1}$ which converges in $C^2(B_i)$. Let $u_i(x)=\lim_{j\to\infty}u_{\delta_{ij}}$, then $u_i$ is a solution to equation (\ref{gove22}) in $B_i$ and $u_i|_{B_{i-1}}=u_{i-1}$. Moreover, $\{u_{\delta_{nn}}\}$ is a subsequence of $u_{\delta_{ij}}$ for every $i$, and $\lim_{n\to\infty}u_{\delta_{nn}}=u_i$ in $B_i$. Let $u=\lim_{n\to\infty}u_{\delta_{nn}}$, then $u$ is a solution of the equation (\ref{gove22}) on $\mathbb R^2$. Since $0\le u\le -u_0$, we have $u$ converges to zero at the infinity. As a consequence, $v=u+u_0$ is a solution to equation (\ref{gove2}). This completes the proof of theorem \ref{cose}.
\end{proof}

\subsection{Asymptotic estimates}
In this section we discuss the asymptotic behavior of the solution $v$. We first consider the conformal factor 
\begin{equation}
\mathrm e^{\eta}=\mathrm e^{16\pi Gc-8\pi GF(\mathrm e^v)+8\pi Gv}\Big(\prod_{s=1}^N |x-p_s|^{-2}\Big)^{8\pi G}
\end{equation}
since $v\rightarrow0$ as $|x|\rightarrow+\infty$, asymptotically we have 
\[\mathrm{e}^{\eta}=O(|x|^{-16\pi GN}) \quad\quad \text{as} \quad|x|\rightarrow+\infty \]
For solution $v$ to the equation (\ref{cos22}), and $|\nabla v|^2$, we have the following decay estimates:

\begin{theorem}
If $8\pi GN<1$, then the solution $v$ to equation (\ref{cos22}) has the bound
\[|v(x)|<C_1|x|^{-b}\quad\quad |x|>r_0\]
\[|\nabla v|^2<C_2|x|^{-b}\quad\quad |x|>r_0\]
for any $b>0$ and some large $r_0>0$, the values of constants $C_1$, $C_2$ depend on b. Moreover, we have
\[|\phi|^2-1=O(|x|^{-b}),\quad F_{12}=O(|x|^{-b}),\quad |D_\mu\phi|=O(|x|^{-b}),\quad\mu=1,2\]
and the quantized energy and total curvature
\[\int_{\mathbb R^2} E=\int_{\mathbb R^2}F_{12}=2\pi N\quad \int _{\mathbb R^2}K_g\mathrm d \Omega_g=16\pi^2GN\]
\end{theorem}
\begin{proof}
For the prescribed zeros $p_i$, we can choose $r_0>0$ suitably large such that $\{p_1,p_2,...,p_N\}\subset\{x\in\mathbb R^2:x<r_0\}$. Let $\Omega=\mathbb R^2\setminus\overline{B(r_0)}$, we have
\begin{equation}\label{asy}
\Delta v=\frac{2\mathrm{e}^{\eta}U(\mathrm{e}^v)}{\sqrt{G(\mathrm{e}^v)}\sqrt{1-\frac{1}{b^2}U^2(\mathrm{e}^v)}}\quad\quad x\in\Omega
\end{equation}
Choose comparison function $w_1(x)=C|x|^{-b}$, then we have
\begin{equation}
\Delta w_1=b^2|x|^{-2}w_1 \quad |x|>r_0
\end{equation} 
Therefore, when $|x|>r_0$, we have
\begin{align}
\Delta (v+w_1)&=\frac{2\mathrm{e}^{\eta}U(\mathrm{e}^v)}{\sqrt{G(\mathrm{e}^v)}\sqrt{1-\frac{1}{b^2}U^2(\mathrm{e}^v)}}+b^2|x|^{-2}w_1\notag\\
&=\frac{2\mathrm{e}^{\eta}U'(\mathrm{e}^\xi)\mathrm{e}^\xi v}{\sqrt{G(\mathrm{e}^v)}\sqrt{1-\frac{1}{b^2}U^2(\mathrm{e}^v)}}+b^2|x|^{-2}w_1\notag
\end{align}
where $v<\xi<0$. 
Since $\mathrm{e}^{\eta}=O(|x|^{-16\pi GN})$ when $|x|\rightarrow+\infty$, and $16\pi GN<2$, when $|x|>r_0$ is large enough, we have
\[\frac{2\mathrm{e}^{\eta}U'(\mathrm{e}^\xi)\mathrm{e}^\xi v}{\sqrt{G(\mathrm{e}^v)}\sqrt{1-\frac{1}{b^2}U^2(\mathrm{e}^v)}}<b^2|x|^{-2}v\]
Therefore
\begin{equation}
\Delta (v+w_1)<b^2|x|^{-2}(v+w_1) \quad |x|>r_0
\end{equation}
On the other hand, we can always find $C>0$ such that $(v+w_1)|_{|x|=r_0}>0$, thus we have $v+w_1>0$, $\forall x\in\mathbb R^2\setminus\overline{B(r_0)}$ by maximum principle. Therefore, for any $b>0$, we can find $C_b>0$ such that $-C_b|x|^{-b}<v(x)<0$ when $x>r_0$.

Similarly, we can get the the asymptotic behavior of $|\nabla v|^2$. Since $v\in L^2(\Omega)$, from equation (\ref{asy}) and $L^2$-estimates we have $v\in W^{2,2}(\Omega)$. By differentiating equation (\ref{asy}), we get
\begin{equation}\label{asy2}
\begin{aligned}
\Delta(\partial_jv)&=2\mathrm{e}^{\eta}\frac{2U'(\mathrm{e}^v)G(\mathrm{e}^v)\mathrm{e}^v-U(\mathrm{e}^v)G'(\mathrm{e}^v)\mathrm{e}^v+b^{-2}U^3(\mathrm{e}^v)G'(\mathrm{e}^v)\mathrm{e}^v}{2\Big(G(\mathrm{e}^v)(1-b^{-2}U^2(\mathrm{e}^v))\Big)^{\frac{3}{2}}}\partial_jv\\
&\quad+\frac{2\mathrm{e}^{\eta}U(\mathrm{e}^v)}{\sqrt{G(\mathrm{e}^v)}\sqrt{1-\frac{1}{b^2}U^2(\mathrm{e}^v)}}\partial_j\eta\\
&=2\mathrm{e}^{\eta}\frac{2U'(\mathrm{e}^v)G(\mathrm{e}^v)\mathrm{e}^v-U(\mathrm{e}^v)G'(\mathrm{e}^v)\mathrm{e}^v+b^{-2}U^3(\mathrm{e}^v)G'(\mathrm{e}^v)\mathrm{e}^v}{2\Big(G(\mathrm{e}^v)(1-b^{-2}U^2(\mathrm{e}^v))\Big)^{\frac{3}{2}}}\partial_jv\\
&\quad+\frac{2\mathrm{e}^{\eta}U(\mathrm{e}^v)}{\sqrt{G(\mathrm{e}^v)}\sqrt{1-\frac{1}{b^2}U^2(\mathrm{e}^v)}}\Big(8\pi G(-F'(\mathrm{e}^v)\mathrm{e}^v+1)\partial_jv-8\pi G\partial_j\Big(\ln\Big(\prod_{s=1}^N|x-p_s|^2\Big)\Big)\Big)\\
&=\Big(\Big(\frac{2U'(\mathrm{e}^v)G(\mathrm{e}^v)\mathrm{e}^v-U(\mathrm{e}^v)G'(\mathrm{e}^v)\mathrm{e}^v+b^{-2}U^3(\mathrm{e}^v)G'(\mathrm{e}^v)\mathrm{e}^v}{\Big(G(\mathrm{e}^v)(1-b^{-2}U^2(\mathrm{e}^v))\Big)^{\frac{3}{2}}}\Big)\\
&\quad+\frac{U(\mathrm{e}^v)}{\sqrt{G(\mathrm{e}^v)}\sqrt{1-\frac{1}{b^2}U^2(\mathrm{e}^v)}}16\pi G(-F'(\mathrm{e}^v)\mathrm{e}^v+1)\Big)\mathrm{e}^{\eta}\partial_jv\\
&\quad-\frac{16\pi GU(\mathrm{e}^v)}{\sqrt{G(\mathrm{e}^v)}\sqrt{1-\frac{1}{b^2}U^2(\mathrm{e}^v)}}\mathrm{e}^{\eta}\partial_j\Big(\ln\Big(\prod_{s=1}^N|x-p_s|^2\Big)\Big)
\end{aligned}
\end{equation}
The RHS of (\ref{asy2}) is also in $L^2(\Omega)$, therefore we have $\partial_{j}v\in W^{2,2}(\Omega)$. Moreover, asymptotically we have 
\[\Delta(\partial_jv)=M|x|^{-16\pi GN}\partial_jv+\alpha(x)|x|^{-\alpha}\]
where $M$ and $\alpha$ are large constant, and $\alpha(x)$ is a bounded function. Therefore, asymptotically we have
\[\Delta (\partial v)^2\ge C|x|^{-16\pi GN}(\partial_jv)^2-c_\alpha|x|^{-\alpha}\]
for some large $C$ and $\alpha$. Choose comparison function $w_2(x)=C_2|x|^{-2b}$, then 
\[\Delta w_2=4b^2w_2|x|^{-2}\]
When $|x|$ is large enough we have,
\begin{equation}
\begin{aligned}
\Delta\Big((\partial_jv)^2-w_2\Big)&\ge C|x|^{-16\pi GN}(\partial_jv)^2-4b^2w_2|x|^{-2}-c_{\alpha}|x|^{-\alpha}\\
&=C|x|^{-16\pi GN}(\partial_jv)^2-5b^2w_2|x|^{-2}+(b^2w_2|x|^{-2}-c_{\alpha}|x|^{-\alpha})
\end{aligned}
\end{equation}
For any $b>0$, we can find suitably large $\alpha>0$ and choose $C_2>0$ such that \[b^2w_2|x|^{-2}-c_{\alpha}|x|^{-\alpha}>0\] Therefore, for some large $r_0>0$, we have
\begin{equation}
\Delta\Big((\partial_jv)^2-w_2\Big)\ge 5b^2|x|^{-2}\Big((\partial_jv)^2-w_2\Big) \quad x\in \Omega
\end{equation}
Moreover, when $C_2$ is large enough, we have $((\partial_jv)^2-w_2)|_{|x|=r_0}<0$. Thus, by maximum principle we get $(\partial_jv)^2\le w_2$ when $|x|>r_0$. Therefore, $|\nabla v|^2<C_2|x|^{-b}$ when $|x|>r_0$. As a consequence, from formulae (\ref{fml}), we have
 \[|\phi|^2-1=O(|x|^{-b}),\quad F_{12}=O(|x|^{-b}),\quad |D_\mu\phi|=O(|x|^{-b}),\quad\mu=1,2\]
Similar to the vortex solution case, we can again use the decay estimates to get the quantized energy
\begin{equation}
E=\int_{\mathbb R^2}\mathcal H\mathrm d\Omega_g=\int_{\mathbb R^2} F_{12}=2\pi N
\end{equation}
From the Einstein equation (\ref{cos3}), we get 
\begin{equation}
\int_{\mathbb R^2}K_g\mathrm d\Omega_g=16\pi^2GN
\end{equation}
\end{proof}

In conclusion, we have established the existence of the vortex solution arising in the generalized Born-Infeld model (\ref{ac}) in a general class of modifying functions $U(t)$, $G(t)$, and $w(t)$ and extended the results in \cite{gbf, han16}. The solution satisfies the Euler-Lagrange equations (\ref{eq1}), (\ref{eq2}) of the model. It is interesting to note that the condition $U(1)=0$ plays an important role in these results. Moreover, we have shown the existence of cosmic strings when the generalized Born-Infeld system is coupled with the Einstein equations. Different from the vortex solutions, there is an upper bound for the number of strings. In both cases, the magnetic flux, minimal energy, and total curvature (for cosmic stings) are proportional to the number of vortices and strings, respectively.

There are several directions that can be further explored. An unsettled problem in the Born-Infeld model is whether the solutions of the Euler-Lagrange equations (\ref{eq1}) and (\ref{eq2}) are also the solutions of the self-dual system (\ref{mainbps1}), (\ref{mainbps2}). One can also investigate the minimizing solution of the Born-Infeld energy in the case when the potential function $V(|\phi|)$ is beyond the special form (\ref{as0}). Moreover, inspired by recent studies in \cite{yang21, yang22}, it is interesting to study the Born-Infeld type model coupled with the Einstein equations and related black hole problems.

\textbf{Acknowledgment.} The author would like to thank Professor Yisong Yang, Professor Deane Yang, and Professor Edward Miller for many helpful discussions and advice.

\bibliographystyle{amsplain}

\end{document}